\documentclass[10pt,a4paper]{article}
\usepackage{amsthm}
\newcommand{\QED}{\qed}
\usepackage{hyperref}

\usepackage{amsmath,amssymb,amsfonts,dsfont}
\usepackage{color}
\usepackage{enumerate}
\usepackage{graphicx,psfrag,subfigure}
\graphicspath{{./}}

	\newcommand{\R}{\mathbb{R}}
	\newcommand{\E}{\mathbb{E}}
	\newcommand{\G}{\mathbb{G}}

	\newcommand{\1}{\mathbf{1}}
	\newcommand{\0}{\mathbf{0}}
	\newcommand{\xx}{\mathbf{x}}
	
	\newcommand{\nn}{\mathbf{n}}
	\newcommand{\bb}{\mathbf{b}}

	\newcommand{\Inci}{A}
	
	\newcommand{\Traccia}{\mathop{\mathrm{tr}}}

	\newtheorem{theorem}{Theorem}
	
	\newtheorem{proposition}[theorem]{Proposition}
	
	\newtheorem{remark}{Remark}
		\newtheorem{assumption}{Assumption}
	
	\newcommand{\Htn}{H_t}  

\newcommand{\ww}{\mathbf{w}}

\providecommand{\com[2]}{\begin{tt}[#1: #2]\end{tt}}

\newcommand{\fromto}[2]{\{#1,\dots, #2\}}

\newcommand{\subscr}[2]{#1_{\textup{#2}}}

\newcommand{\setdef}[2]{\left\{#1 \, : \; #2\right\}}

\newcommand{\degmax}{\subscr{d}{max}}
\newcommand{\eps}{\varepsilon}

\title{Limited benefit of cooperation in distributed relative localization}

\author{Wilbert Samuel Rossi \and Paolo Frasca \and Fabio Fagnani
\thanks{The authors are with the Dipartimento di Scienze Matematiche (DISMA), Politecnico di Torino, corso Duca degli Abruzzi 24, 10129 Torino, Italy. E-mail contacts: {\tt \small \href{mailto:wilbertsamuel.rossi@polito.it}{wilbertsamuel.rossi@polito.it} }, {\tt\small \href{mailto:paolo.frasca@polito.it}{paolo.frasca@polito.it}}, {\tt\small \href{mailto:fabio.fagnani@polito.it}{fabio.fagnani@polito.it}}. 
}
}

\begin{document}

\maketitle
\thispagestyle{empty}
\pagestyle{empty}

\begin{abstract}
Important applications in robotic and sensor networks
require distributed algorithms to solve the so-called
relative localization problem: a node-indexed vector has to
be reconstructed from measurements of differences between
neighbor nodes. In a recent note, we have studied the estimation
error of a popular gradient descent algorithm showing that the mean
square error has a minimum at a finite time, after which
the performance worsens.
This paper proposes a suitable modification of this algorithm incorporating more realistic {\it a priori} information on the position. The new algorithm presents a performance monotonically decreasing to the optimal one. Furthermore, we
show that the optimal performance is approximated, up to a $1+\eps$ factor, within a time which is independent of the graph
and of the number of nodes. This convergence time is very much related to the minimum exhibited by the previous algorithm and both lead to the following conclusion: in the presence of noisy data, cooperation is only useful till a certain limit.
\end{abstract}


\section{Introduction}
\label{sect:intro}
We study in this paper the distributed solution of a problem of relative localization in a network of sensors. 
We assume to have a group of agents organized in a graph and a vector, indexed over the agents and unknown to them: the agents are allowed to take relative noisy measurements of their vector entries with respect to their neighbors in the graph. The estimation problem consists in reconstructing the original vector, up to an additive constant. We refer to this problem as the problem of {\em relative localization}.


\subsection*{Contribution}
In our previous work~\cite{WSR-PF-FF:12}, we studied the performance of a distributed algorithm, obtained as a gradient descent solution after of a least-squares formulation of the localization problem. The mean square estimation error of this algorithm has a minimum at a finite time, after which the performance worsens. 
This non-monotonic behavior, although very interesting from a theoretical point of view, may be seen as a potential drawback of the algorithm. For this reason, in the present paper we build on the insights gained from our previous work to present an algorithm with monotonic mean square error performance.

As the main contribution of this work, we define an $\eps$-convergence time for the algorithm and we find an upper bound on it, which has the remarkable feature of being independent of the network and even of number of sensors. Notably, also the minimum time of the algorithm in~\cite{WSR-PF-FF:12} has an upper bound which is independent of the graph. Both these observations suggest that cooperation provides limited benefit in reconstructing estimates from measurements which are affected by noise. Indeed, a bounded optimal time means that there is no advantage for a node in obtaining data from outside a certain neighborhood.
Intuitively, communication with sensors which are far away in the network does not contribute enough significant information: then, the noise which corrupts the data makes it useless (in the algorithm below) or even misleading (in the more na\"ive algorithm in~\cite{WSR-PF-FF:12}).

\subsection*{Related Literature}
The problem of relative localization has been brought to our attention in the formulation of~\cite{PB-JPH:07,PB-JPH:08,PB-JPH:09}, which is slightly different from ours, as these authors assume to have an anchor node, in order to avoid the uncertainty about the additive constant.
The natural applications of this estimation problem include spacial localization and clock synchronization~\cite{AG-PRK:06a,RC-AC-LS-SZ:11,NMF-SRG-PRK:11}. Distributed algorithms have been proposed in several papers, including~\cite{PB-JPH:07,AG-PRK:06a,SB-SDF-LS-DV:10}, and contemporary work is focusing on randomized algorithms~\cite{NMF-AZ:12,CR-PF-HI-RT:13a,RC-LS:12}.

\subsection*{Paper organization}
In Section~\ref{sect:relative-localization} we define the problem of relative localization, and our novel algorithm for its solution is derived in Section~\ref{sect:algo-def}. Then in Section~\ref{sect:analysis} we analytically study the convergence and the mean square error of the algorithm, while simulations are described in Section~\ref{sect:simulations}. We conclude with a short section which summarizes our contribution and points to future research.

\subsection*{Notation}
Vectors are denoted with boldface letters, and matrices with capital letters. By the symbols $\1$ and $\0$ we denote vectors having all entries equal to $1$ and $0$, respectively. Given a matrix $M$, we denote by $\Traccia(M)$ its trace, by $M^\top$ its transpose and by $M^\dagger$ its Moore-Penrose pseudo-inverse.

\section{The relative localization problem}\label{sect:relative-localization}
We consider a set of $N$ agents, and we endow each of them with a scalar quantity  $\bar{x}_i \in \R$, for $i\in\fromto{0}{N-1}$. The $i$th agent does not know the value $\bar x_i$, but has an estimate $x_i \in \R$. We shall denote by $\bar{\xx}$ and $\xx$ the $N$-dimensional vectors whose components are $\bar x_i$ and $x_i$, respectively.
We suppose that each agent $i$ can take relative measurements $\bar x_i-\bar x_j$ with respect to some neighbors $j$. 
An undirected graph $\G = (\fromto{0}{N-1},E)$ is used to represent the available measurements. The set of vertices is constituted by the $N$ agents, and the edges (pairs of agents) in $E$ correspond to the available measurements. We assume that there are $M$ available measurements, and that measurements are symmetrical, meaning that both agents of a pair know the measurement, with a reversed sign. Furthermore, we assume that the graph $\G$ is connected.
On each edge, we choose an orientation, that is, we define a starting node and an ending node, in order to encode the measurements by using the incidence matrix $\Inci \in \R^{M \times N}$ defined as follows
	\begin{displaymath}
		(\Inci)_{e,i} = \left\{ \begin{array}{ll}
			\phantom{+}1 & \textrm{if $i$ is the terminating edge of $e$ } \\
			-1 & 			\textrm{if $i$ is the starting edge of $e$ } \\
			\phantom{+}0 & \textrm{otherwise.}
			\end{array}  		\right.
	\end{displaymath} 
Measurements are affected by errors, which can be modeled by independent and identically distributed noises. Let $\bb \in \R^M$ be the vector of the measurements and $\nn \in \R^M$ that of noises. Then, in  matrix notation we have
	\begin{align*}
		\bb = A \bar{\xx} + \nn
	\end{align*}
with $\E\left[\nn\right] = \0$ and $\E\left[\nn \nn^\top\right]= \sigma^2 I$
where $I \in \R^{M\times M}$ is the identity matrix.	
It is also useful to define the Laplacian of $\G$ as $L = \Inci^\top \Inci$. The Laplacian $L$ is a symmetric matrix, and being $\G$ connected, $L$ has eigenvalues $\lambda_0 = 0 $ and $0 < \lambda_i \leq 2 \degmax$ for $i\in\fromto{1}{N-1}$, with $\degmax$ denoting the maximum degree of the nodes.

	
\section{Definition of the algorithm } 
\label{sect:algo-def}
In view of the statistical assumptions on the noise affecting the measurements, 
a natural approach to the relative localization problem involves solving the least-squares problem
$$\min_{z}\| \Inci \xx - \bb \|_2^2.$$
This approach has already been taken in the literature, and leads to design the distributed algorithm studied in~\cite{WSR-PF-FF:12}.
In this paper, we additionally assume that each node $i$ has an {\it a priori} information on $\bar \xx$, which is known to be a random vector independent from $\nn$ and such that $\E\left[\bar{\xx}\right]=  \xx_0$ and $\E\left[(\bar{\xx} - \xx_0)(\bar{\xx} - \xx_0)^\top\right]= \nu^2 I$.
In order to exploit this statistical information, we choose to minimize the functional
	$$ \Phi(\xx) = \frac{1}{\sigma^2}\| A \xx - \bb \|_2^2 + 
	\frac{1}{\nu^2} \| \xx- \xx_0  \|_2^2,$$
which includes both the information obtained by the measurements and the {\it a priori} knowledge about $\bar \xx$, weighted according to their significance, {\it i.e.}, the inverse of their variances. Compared to $\Psi(\xx)$, the extra term in this functional can be seen as a Tikhonov regularization term, which turns the estimation problem at hand into a problem of maximum a posteriori probability (MAP) estimation. We refer the reader to~\cite[\S6.3.2 and \S7.1.2]{SB-LV:04} for a broad introduction to these concepts.

As $\Phi$ is convex, it is natural to consider gradient descent algorithms for its minimization.
Provided we define $\gamma = \frac{\sigma^2}{\nu^2}$, the gradient of the objective function is $\nabla \Phi(x)= \frac{2}{\sigma^2} \big(A^\top A \xx - A^\top \bb + \gamma(\xx  - \xx_0)\big)$ , 
so that a gradient descent iterate can be defined as 
\begin{align*}
	\xx[t+1] &= \xx[t] - \tau \frac{2}{\sigma^2}\nabla \Phi \left(\xx[t]\right) \\
			 &= \xx[t] - \tau \left[ \left(A^\top A \xx[t] - A^\top \bb \right) 
			 		+ \gamma \left( \xx[t] - \xx_0 \right) \right] \\
			 &= (I - \tau L - \tau \gamma I ) \xx[t] + \tau A^\top b + \tau \gamma \xx_0 
\end{align*}
for a suitable $\tau>0$. Equivalently, we may write the algorithm as
\begin{equation} \label{eq:gradient-algo}
	\left\{ \begin{array}{l}
		\xx[t+1] = Q \xx[t] + \ww \\
		\xx[0] = \xx_0
		\end{array} \right.
\end{equation}
where 
\begin{align}\label{Q_def}
Q = I - \tau L - \tau \gamma I = (1-\tau \gamma)I-\tau L
\end{align}
and 
$$\ww = \tau A^\top \bb + \tau \gamma \xx_0.$$

Remarkably, this algorithm is distributed, in the following sense. The matrix $Q$ is adapted to the graph $\G$, {\it i.e,}, $Q_{ij}=0$ if $(i,j)\not \in E$: then, in order to update a component  as
$x_i[t+1] = \sum_j Q_{ij} x_j[t] + w_i $, 
the algorithm requires communication and measurements only with the nodes which are neighbors of $i$ in the graph.

\section{Analysis}\label{sect:analysis}
In the analysis of algorithm~\eqref{eq:gradient-algo} and from here on in this paper, we shall make the following assumption, which is sufficient to our results.
\begin{assumption}\label{ass:on-tau}
The graph $\G$ is connected and $$\tau \leq \frac{1}{\degmax+ \gamma}.$$
\end{assumption}

\subsection{Stability properties}
We begin our analysis by studying the convergence properties of the proposed algorithm.
\begin{proposition}[Convergence]\label{prop:convergence}
	If Assumption~\ref{ass:on-tau} is satisfied, then the algorithm~\eqref{eq:gradient-algo} 
	converges at exponential rate to
	$$\xx^* = (A^\top A + \gamma I)^{-1}(A^\top \bb +\gamma\xx_0),$$
	which is the optimal solution to the problem 
	$$\min_{\xx}\Phi(\xx)$$
\end{proposition}
\begin{proof}
First, we show that $\xx^*$ is the optimal solution of the optimization problem. 
To this goal, we equate to the zero of the gradient $\nabla \Psi(\xx)$ and solve the normal equation
	\begin{align*}
		(A^\top A + \gamma I) \xx &= A^\top \bb + \gamma \xx_0 
	\end{align*}
Since $\gamma >0$, the matrix $A^\top A + \gamma I = L + \gamma I$ is invertible, and hence the optimal solution is unique and equal to $ \xx^* $.
	
Second, we show that the algorithm converges to $\xx^*$ . 
By solving the recursion we have
	\begin{align} \label{alg_explicit}
		\xx[t] &= Q^t \xx_0 + \sum_{n=0}^{t-1} Q^n \ww
	\end{align}
Since $Q = (1-\tau \gamma) I - \tau L$, also $Q$ is diagonalizable with real eigenvalues
	$ \xi_i = 1 - \tau \gamma - \tau \lambda_i $. Using Assumption~\ref{ass:on-tau}, W
we have
	\begin{align*}
		\max_i \xi_i &= 1 - \tau \gamma - \tau \min_i \lambda_i \\
					&= 1 - \tau \gamma - \tau \lambda_0 \\&= 1 - \tau \gamma = \xi_0 \\
		\min_i \xi_i &= 1 - \tau \gamma - \tau \max_i \lambda_i \\
					&\geq  1 - \tau \gamma - \tau 2 \degmax \\
					&= 2 - 2\tau \gamma - 2 \tau \degmax -1 + \tau \gamma.
	\end{align*}	  
	Therefore, given the assumptions on $\tau$, all the eigenvalues of $Q$ belong to the interval 
	$ [-1+\tau\gamma , 1-\tau \gamma]$ (note that $\tau \gamma <1$), 
	the algorithm is exponentially convergent,
	and $ \lim_{t \to \infty} Q^t = 0.$
%
	Then, we can compute  
\begin{align*}
	\xx[t] &= Q^t \xx_0 + (I - Q)^{-1} (I - Q) \sum_{n=0}^{t-1} Q^n \ww \\
	 &= Q^t \xx_0 + (I - Q)^{-1} (I - Q^t) \ww
\end{align*}	
and consequently
\begin{align*}
		\lim_{t \to \infty} \xx[t] &= (I - Q ) \ww \\
		&= (L + \gamma I)^{-1} (A^\top \bb +\gamma\xx_0) = \xx^*.
	\end{align*}	
	
\end{proof}

\begin{remark}[Average preservation]
	The algorithm preserves the barycenter (or average) of the state, namely $\frac{1}{N} \1^\top \xx[t] = \frac{1}{N} \1^\top \xx_0 $. 
Remarkably, this property holds even if $Q$ is not stochastic. 
Notice indeed that $\1^\top \ww = \tau \gamma \1^\top \xx_0$ and that $\1^\top \xx[t+1] = (1-\tau\gamma) \1^\top \xx[t] + \tau \gamma \1^\top \xx_0$. Since $\xx[0] = \xx_0$, by induction the barycenter is preserved.
This property is also consistent with the intuition that the optimal solution must satisfy
$\frac{1}{N} \1^\top \xx^* = \frac{1}{N} \1^\top \xx_0 $. \hfill\QED
\end{remark}

\subsection{Transient mean-square performance}
To evaluate the algorithm performance, we follow the approach in~\cite{FG-SZ:11} 
and define the performance metric as the mean square error between the current estimate $\xx[t]$ and the true configuration $\bar{\xx}$, that is,
$$\Htn := \frac{1}{N}	\E \| \xx[t] - \bar{\xx} \|_2^2,$$
where the expectation is taken on both the noise $\nn$ and the initial condition $\xx_0$. 
This performance metric can be computed in terms of the eigenvalues of the matrix $Q$.
	
\begin{proposition}[Mean square performance]\label{prop:time-performance}
If Assumption~\ref{ass:on-tau} is satisfied, then the following equality holds
	\begin{align*}
		\Htn= \frac{\nu^2}{N} \sum_{i=0}^{N-1} \xi_i^{2t} + 
		   \frac{\tau \sigma^2}{N} \sum_{i=0}^{N-1} \frac{1- \xi_i^{2t}}{1-\xi_i},	
	\end{align*}
	where $\xi_i$'s are the eigenvalues of $Q$.
\end{proposition}
\begin{proof}
	We express $\ww$ in terms of $\nn$ and $\xx_0 - \bar{\xx}$ as
	\begin{align*}
		\ww &= \tau A^\top A \bar{\xx} + \nn + \tau \gamma \xx_0 \\
		    &= \left(  I -\tau \gamma I  - Q \right) \bar{\xx} + \tau A^\top \nn + \tau \gamma \xx_0 \\
		    &= \left(  I  - Q \right) \bar{\xx} + \tau A^\top \nn + \tau \gamma (\xx_0 - \bar{\xx})
	\end{align*}
Now, we compute $\xx[t] - \bar{\xx}$, given $\ww$ and \eqref{alg_explicit} as 
	\begin{align*}
		\xx[t] - \bar{\xx} =  Q^t (\xx_0 - \bar{\xx}) & + 
		          \tau \gamma \sum_{n=0}^{t-1} Q^n (\xx_0 - \bar{\xx}) \\
		          & +    \tau \sum_{n=0}^{t-1} Q^n A^\top \nn. 
	\end{align*}
	From the definition of $\Htn$ we have
	\begin{align*}
		\Htn = \frac{1}{N}	\E\left[\Traccia\left[(\xx[t] - \bar{\xx})(\xx[t] - \bar{\xx})^\top \right] \right]
	\end{align*}
By using  the above formula for $\xx[t] - \bar{\xx}$, we get
	\begin{align*}
		\Htn =& \frac{1}{N} \Traccia \left[ \nu^2 Q^{2t} + 2 \tau \gamma \nu^2 \sum_{m=0}^{t-1} Q^{t+m} + \right. \\
		&+\left. \tau  \sum_{n=0}^{t-1} \sum_{m=0}^{t-1} Q^{n+m} ( \tau \sigma^2 L + \tau \gamma^2 \nu^2 I) \right],
	\end{align*}
through some algebraic manipulations --which we omit-- involving the properties of the trace operator, the linearity of expectation and the symmetry of $Q$. 
	Now, given that $\gamma = \frac{\sigma^2}{\nu^2}$, we obtain	\begin{align*}
		\Htn &= \frac{1}{N} \Traccia \left[ \nu^2 Q^{2t} + \tau \sigma^2 (I + Q^t) \sum_{n=0}^{t-1} Q^n \right] \\
		     &= \frac{1}{N} \Traccia \left[ \nu^2 Q^{2t} + \tau \sigma^2 (I + Q^t) (I - Q^t) (I-Q)^{-1} \right] \\
		     &= \frac{1}{N} \Traccia \left[ \nu^2 Q^{2t} + \tau \sigma^2 (I - Q^{2t})(I-Q)^{-1} \right].
	\end{align*}
	Notice that the matrix $(I-Q)$ is invertible since it is proportional to $L+\gamma I$. The result follows immediately as $\xi_i$s are the 
	eigenvalues of $Q$.
\end{proof}

The key property of monotonicity of $\Htn$ is stated in the next result.
\begin{theorem}[Monotonicity of $\Htn$]
If Assumption~\ref{ass:on-tau} is satisfied, then $\Htn$ is strictly decreasing and $$H_\infty:=\lim_{t\to+\infty} \Htn =\frac{\tau \sigma^2}{N} \sum_{i=0}^{N-1} \frac{1}{1-\xi_i}.$$
\end{theorem}
\begin{proof}	
	Let us recall the definition of $\gamma$ and define a new constant $\alpha$, 
	according to what done in~\cite{WSR-PF-FF:12}:
	\begin{align}\label{eq:define-alpha}
		\gamma = \frac{\sigma^2}{\nu^2} \quad \quad 
		\alpha = \frac{\nu^2}{\tau \sigma^2} = \frac{1}{\tau \gamma}
	\end{align}	
	Note that, given Assumption~\ref{ass:on-tau}, $\alpha > 1 + \degmax \frac{\nu^2}{\sigma^2}$. 
Keeping this inequality in mind we can rewrite $\Htn$ as 
	\begin{align*}
		\Htn = \frac{\tau \sigma^2}{N} \sum_{i=0}^{N-1} \left[ \alpha \xi_i^{2t} +  				\frac{1- \xi_i^{2t}}{1-\xi_i} \right].
	\end{align*}
	We will show that $\Htn$ is  decreasing in $t$, since the $i^{\textup{th}}$ 
	term in the sum is either a constant or decreasing sequence. Let us compute the finite increment
	\begin{align*}
		H_{t+1} - \Htn & =  \frac{\tau \sigma^2}{N} \sum_{i=0}^{N-1} \left[ \xi_i^{2t} 
				\left( \alpha \xi_i^2  - \alpha +1 + \xi_i \right) \right] \\
				& = \frac{\tau \sigma^2}{N} \sum_{i=0}^{N-1} \left[ \xi_i^{2t} h(\xi_i) \right],
	\end{align*}
	with $h(\xi) = \alpha \xi^2 + \xi + 1 - \alpha$.
	Note that $\xi^{2t} > 0 $ for all $\xi \neq 0$  whereas 
	$h(\xi) < 0$ when $\xi \in (-1 , 1- \frac{1}{\alpha}) = (-1 , 1-\tau \gamma)$.
	Since $\xi_i \in [-1+\tau \gamma , 1- \tau \gamma)$ when $i>0$, the corresponding contribution in $\Htn$ is a decreasing sequence (unless $\xi_i = 0$). 
	The contribution of $\xi_0 = 1 - \tau \gamma$ in $\Htn$ is constant, since
	$\alpha \xi_0^{2t} + \frac{1 - \xi_0^{2t}}{\tau \gamma} = \alpha$. 
	This corresponds to the invariance of the barycenter. 

	The sequence $\Htn$ is bounded and monotonic, so it has a limit that we 	can also compute explicitly as
	\begin{align*}
		 \lim_{t\to+\infty} \Htn &
		       = \frac{\tau \sigma^2}{N} \sum_{i=0}^{N-1} \frac{1}{1-\xi_i} \\
		       & = \frac{\sigma^2}{N} \sum_{i=0}^{N-1} \frac{1}{\gamma + \lambda_i},
	\end{align*}
	where $\lambda_i$ are the eigenvalues of the Laplacian of the graph.
\end{proof}

\begin{remark}[Meaning of $H_\infty$]
It is worth to recall that the asymptotical error, which we can also write as $H_\infty= \frac{1}{N}	\E \| \xx^* - \bar{\xx} \|_2^2$, only depends on the properties of $\xx^*$ as the solution of the regularized least-squares problem: hence it does not depend on the algorithm.
\end{remark}


\subsection{Near-optimal stopping time}

For every $\eps>0$, we can define a near-optimal stopping time, after which the estimation error is only a $(1+\eps)$ factor larger than the optimal one:
$$t^*_\eps=\inf\setdef{t}{\Htn<(1+\eps) H_\infty}.$$
	
The following estimate shows that the algorithm can be stopped, with a guaranteed loss of accuracy with respect to the regularized least-squares optimum, after a time which does not depend the graph or even on the number of sensors. 
\begin{proposition}[Universal bound on stopping time]	\label{prop:bound-t-eps}
If Assumption~\ref{ass:on-tau} is satisfied, then it holds
\begin{equation}
\label{eq:t-eps-estimate}t^*_\eps\le\displaystyle  \frac{\alpha}{2} \log{\left( \frac{2 \alpha}{\eps}\right)},\end{equation}
where $\alpha$ is defined in~\eqref{eq:define-alpha}.
\end{proposition}
\begin{proof}
From the definition we immediately deduce that 
\begin{align*}
t^*_\eps=\inf\setdef{t}{ \sum_{i=0}^{N-1} \bigg(\alpha\xi_i^{2t} + \frac{1- \xi_i^{2t}}{1-\xi_i}\bigg)	<  \sum_{i=0}^{N-1} \frac{1+\eps}{1-\xi_i}
}.
\end{align*}
By taking an upper bound on the second term of the left-hand side of the inequality, we have 
\begin{align*}
t^*_\eps\le\inf\setdef{t}{	 \sum_{i=0}^{N-1} \alpha \xi_i^{2t}<  \sum_{i=0}^{N-1} \frac{\eps}{1-\xi_i}
}.
\end{align*}
Since Assumption~\ref{ass:on-tau} implies $|\xi_i|\le 1-\tau\gamma$ and $\frac1{1-\xi_i}\ge \frac1{2-\tau\gamma}>\frac12$, we have 
\begin{align*}
t^*_\eps\le\inf\setdef{t}{	\alpha (1-\tau\gamma)^{2t}<\frac\eps2}.
\end{align*}
By solving for $t$ in the above inequality we get 
\begin{align*}
t^*_\eps \le \frac{  \log{\left( \frac{2 \alpha}{\eps}\right)} }{  \log\frac1{(1 - \alpha^{-1} )^2} }, 
\end{align*}
and then the result follows.
\end{proof}

\section{Simulations and comparison with~\cite{WSR-PF-FF:12}}\label{sect:simulations}

We have simulated algorithm~\eqref{eq:gradient-algo} and numerically evaluated the related performance metrics, assuming the graph to be a cycle.
Figure~\ref{fig:t-epsilon} compares the simulated and expected performance of the algorithm. Notice that, although the expected error $H_t$ is monotonic, single realizations need not to be monotonic, and indeed some of them show a minimum. 
The figure also shows the actual near-optimal time $t^*_\eps$, in comparison with its estimate obtained in Proposition~\ref{prop:bound-t-eps}. We notice that the estimate is significantly larger than the true value: this looseness is not surprising, as our bound does not exploit  any information about the topology of the sensing and communication graph, which is likely to have a role. Hence, future research may improve upon our bounds by a careful use of information about the spectrum of $Q$, {\it i.e.}, on the graph.

\begin{figure}
\psfrag{stima}{\footnotesize\eqref{eq:t-eps-estimate}}
\psfrag{cross}{\footnotesize$t^*_\eps$}
\psfrag{time}{\footnotesize$t$}
\psfrag{teo-new}{\footnotesize$H_t$}
\psfrag{realizations}{\footnotesize$\frac{1}{N}\|\xx[t]-\bar\xx\|_2^2$}
\includegraphics[width=\columnwidth]{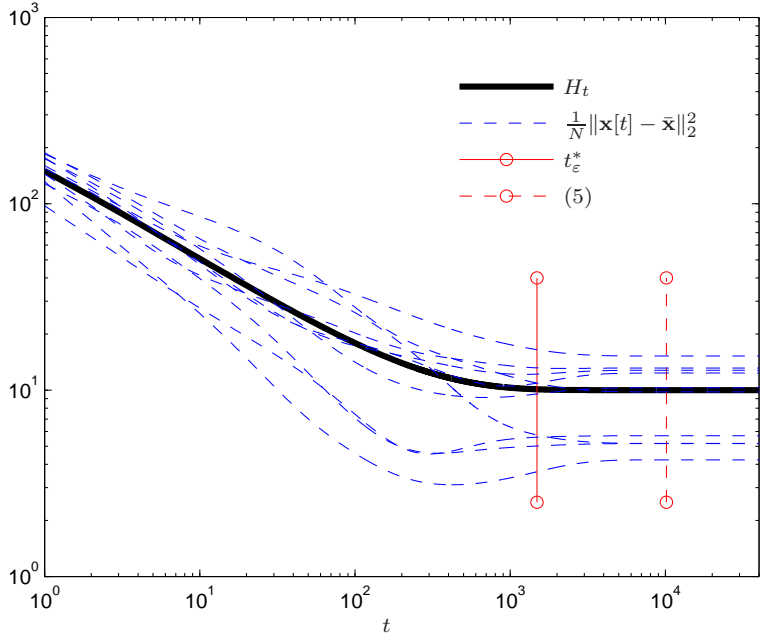}
\caption{Mean square error of algorithm~\ref{fig:t-epsilon} on a cycle graph with $N=160$, $\nu=20$, $\sigma=1$, $\eps=0.01$.}
\label{fig:t-epsilon}
\end{figure}

The second goal of our simulations is to compare algorithm~\eqref{eq:gradient-algo} with the analogous algorithm defined in~\cite[Eq.~(1)]{WSR-PF-FF:12}, based on based on minimizing $\Psi(\xx)$. 
Hence, Figure~\ref{fig:compare-algos} plots for both algorithms the mean square error, together with the mean square error of a few single realizations. We can see that the performance of the two algorithms is roughly similar (in expectation) until the algorithm in~\cite{WSR-PF-FF:12} reaches a time at which its mean square error is minimal. From that time on, the behavior of the two algorithms becomes different, as algorithm~\cite[Eq.~(1)]{WSR-PF-FF:12} accumulates an increasingly larger mean square error, whereas the error of algorithm~\eqref{eq:gradient-algo} decreases further.
We leave to future research a more detailed comparison of the two algorithms, which should include a discussion on the behavior of single realizations, as opposed to the average performance which has been studied so far.

\begin{figure}
\psfrag{time}{\footnotesize$t$}
\psfrag{exp-old}{\footnotesize \cite[Eq.~(1)]{WSR-PF-FF:12} (mean)}
\psfrag{exp-new}{\footnotesize \eqref{fig:t-epsilon} (mean)}

\psfrag{old}{\footnotesize \cite[Eq.~(1)]{WSR-PF-FF:12} (realiz.)}
\psfrag{new}{\footnotesize \eqref{fig:t-epsilon} (realization)}

\includegraphics[width=\columnwidth]{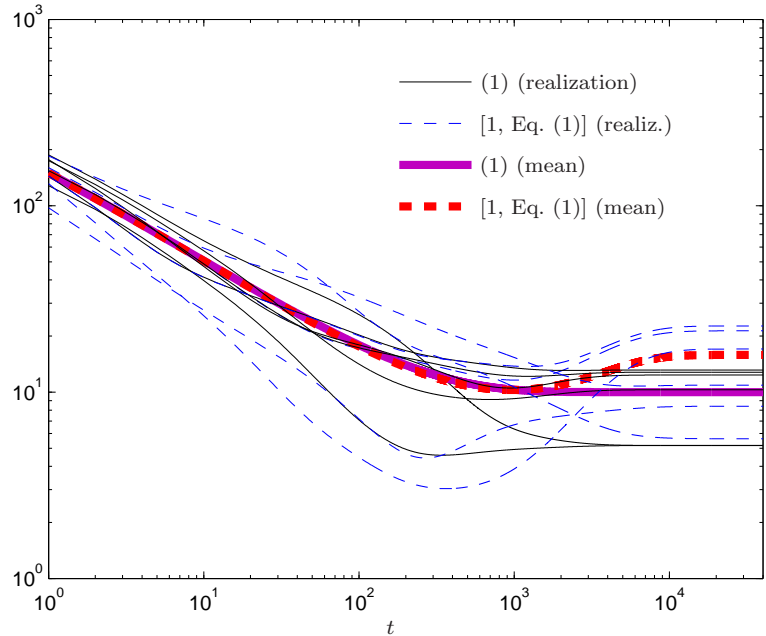}
\caption{Mean square error of algorithms~\eqref{fig:t-epsilon} and~\cite[Eq.~(1)]{WSR-PF-FF:12} on a cycle graph with $N=160$, $\nu=20$, and $\sigma=1$. }
\label{fig:compare-algos}
\end{figure}

\section{Conclusion}
In this paper, we have studied a distributed algorithm to solve the relative localization problem in sensor networks. Compared to algorithms available on literature, the proposed algorithm has an improved performance for large times: moreover, the algorithm is guaranteed to reach (on average) an $\eps$-approximation of the optimal solution within a time which only grows logarithmically in $\eps$ and does not depend on either the topology of the sensor network or the number of sensors. We interpret this feature as an inherent limitation on the benefit of cooperation.
Future research should put our results in a broader context, investigating the fundamental issue of quantifying the benefit of cooperation (if any), depending on the ``cooperation task'' which is assigned to the agents, as well as on the available communication and the measurement models: a recent example of work in this direction is~\cite{AS-JAM-RDA:12}.


%

\end{document}